\documentclass[envcountsame]{llncs}
\pdfoutput=1
\usepackage{times}
\usepackage{helvet}
\usepackage{courier}
\usepackage{url}
\usepackage{graphicx}
\usepackage{microtype}
\usepackage{amssymb}
\usepackage{amsmath}
\usepackage[amsmath]{ntheorem}
\usepackage{float}
\usepackage{subcaption}
\usepackage{cite}

\theoremheaderfont{\itshape}\theorembodyfont{\upshape}
\theoremstyle{nonumberplain}
\theoremseparator{}

\newtheorem{proofsk}{Proof sketch.}

\newcommand{\F}{\ensuremath{\mathcal{F}}}

\newcommand{\bsize}[1]{\ensuremath{|\mathsf{b}({#1})|}}

\newcommand{\Gr}[1]{\ensuremath{\mathcal{G}_{#1}}}
\newcommand{\Span}{\ensuremath{\mathcal{S}}}
\newcommand{\N}{\ensuremath{\mathbb{N}}}
\newcommand{\Horn}{\ensuremath{\mathcal{H}}}
\newcommand{\mgu}{m.g.u.}
\newcommand{\res}{R}

\newcommand{\conn}{\ensuremath{\Horn^c}}
\newcommand{\connmost}[1]{\ensuremath{\conn_{#1,\infty}}}
\newcommand{\chain}{\ensuremath{\Horn^{2c}}}
\newcommand{\chainmost}[1]{\ensuremath{\chain_{#1,\infty}}}
\newcommand{\nocut}{\ensuremath{\Horn^{nr}}}
\newcommand{\cbase}{\ensuremath{C_{base}}}
\newcommand{\cext}[1]{\ensuremath{C_{ext{#1}}}}
\renewcommand{\dots}{..}

\newcommand{\appdx}{}
\newcommand{\conf}[1]{}
\newcommand{\techrep}[1]{#1}

\begin{document}

\title{SLD-Resolution Reduction of Second-Order Horn Fragments\\
  -- technical report --}
\author{Sophie Tourret\inst{1}
\and
Andrew Cropper\inst{2}
}

\institute{
Max Planck Institute for Informatics,
Saarland Informatics Campus, Germany
\email{sophie.tourret@mpi-inf.mpg.de}
\and
University of Oxford, UK\\
\email{andrew.cropper@cs.ox.ac.uk}
}

\maketitle
\begin{abstract}
  We present the \emph{derivation reduction} problem for SLD-resolution, the undecidable problem of finding a finite subset of a set of clauses from which the whole set can be derived using SLD-resolution.
  We study the reducibility of various fragments of second-order Horn logic with particular applications in Inductive Logic Programming.
  We also discuss how these results extend to standard resolution.
\end{abstract}

\setcounter{footnote}{0}
\section{Introduction}

Detecting and eliminating redundancy in a clausal theory (a set of clauses) is useful in many areas of computer science \cite{DBLP:journals/corr/abs-1002-4286,heule2015clause}.
Eliminating redundancy can make a theory easier to understand and may also have computational efficiency advantages \cite{crop:minmeta}.
The two standard criteria for redundancy are entailment \cite{liberatore1,liberatore2,mugg:progol} and subsumption \cite{plotkin:thesis,gottlob1993removing,buntine1988generalized}.
In the case of entailment, a clause $C$ is redundant in a clausal theory $T \cup \{C\}$ when $T \models C$.
In the case of subsumption, a clause $C$ is redundant in a clausal theory $T \cup \{C\}$ when there exists a clause $D\in T$ such that $D$ subsumes $C$.
For instance, consider the clausal theory $T_1$:
\begin{center}
  \begin{tabular}{l}
    $C_1 = p(x) \leftarrow q(x)$\\
    $C_2 = p(x) \leftarrow q(x), r(x)$
  \end{tabular}
\end{center}

\noindent
The clause $C_2$ is entailment and subsumption redundant because it is a logical consequence of $C_1$ (and is also subsumed by $C_1$).
However, as we will soon show, entailment and subsumption redundancy can be too strong for some applications. To overcome this issue, we introduce a new form of redundancy based on whether a clause is \emph{derivable} from a clausal theory using SLD-resolution \cite{sld-resolution}. Let $\vdash^*$ represent derivability in SLD-resolution. Then a Horn clause $C$ is {\em derivationally redundant} in a Horn theory $T \cup \{C\}$ when $T \vdash^* C$.
For instance, in $T_1$, although $C_1$ entails $C_2$, we cannot derive $C_2$ from $C_1$ using SLD-resolution because it is impossible to derive a clause with three literals from a clause with two literals.

We focus on whether theories formed of second-order function-free Horn clauses can be derivationally reduced to minimal (i.e.\ irreducible) finite theories from which the original theory can be derived using SLD-resolution. For instance, consider the following theory $T_2$, where the symbols $P_i$ represent second-order variables (i.e.\ variables that can be substituted by  predicate symbols):
\begin{center}
  \begin{tabular}{l}
    $C_1 = P_0(x) \leftarrow P_1(x)$\\
    $C_2 = P_0(x) \leftarrow P_1(x),P_2(x)$\\
    $C_3 = P_0(x) \leftarrow P_1(x),P_2(x),P_3(x)$\\
    \end{tabular}
\end{center}
Although $C_1$ subsumes $C_2$ and $C_3$, the two clauses cannot be derived from $C_1$ for the same reason as in the previous example.
However, $C_3$ is derivationally redundant because it can be derived by self-resolving $C_2$.
A minimal {\em derivation reduction} of $T_2$ is the theory $\{C_1,C_2\}$ because $C_2$ cannot be derived from $C_1$ and vice versa.

\subsection{Motivation}
Our interest in this form of redundancy comes from Inductive Logic Programming (ILP) \cite{mugg:ilp}, a form of machine learning which induces hypotheses from examples and background knowledge, where the hypotheses, examples, and background knowledge are represented as logic programs. Many forms of ILP \cite{raedt:clint,mugg:metagold,inoue:mil-hex,crop:metafunc,albarghouthi2017constraint,alps} and ILP variants \cite{evans2017learning,andreas,wang2014structure,rock:e2eprolog} use second-order Horn clauses as templates to denote the form of programs that may be induced. For instance, consider the father kinship relation:
\begin{center}
  \begin{tabular}{l}
    $father(A,B) \leftarrow parent(A,B), male(A)$.
  \end{tabular}
\end{center}
A suitable clause template to induce this relation is:
\begin{center}
  \begin{tabular}{l}
    $P_0(A,B) \leftarrow P_1(A,B),P_2(A)$.
  \end{tabular}
\end{center}

Determining which clauses to use for a given learning task is a major open problem in ILP \cite{crop:minmeta,mugg:metagold,crop:thesis}, and most approaches uses clauses provided by the designers of the systems without any theoretical justifications \cite{albarghouthi2017constraint,crop:metafunc,wang2014structure,inoue:mil-hex,rock:e2eprolog,andreas,alps}. The problem is challenging because on the one hand, you want to provide clauses sufficiently expressive to solve the given learning problem. For instance, it is impossible to learn the father relation using only monadic clauses. On the other hand, you want to remove redundant clauses to improve learning efficiency \cite{crop:minmeta}.

To illustrate this point, suppose you have the theory $T_3$:
\begin{center}
  \begin{tabular}{l}
    $C_1 = P_0(A,B) \leftarrow P_1(A,B)$\\
    $C_2 = P_0(A,B) \leftarrow P_1(A,B),P_2(A)$\\
    $C_3 = P_0(A,B) \leftarrow P_1(A,B),P_2(A,B)$\\
    $C_4 = P_0(A,B) \leftarrow P_1(A,B),P_2(A,B),P_3(A,B)$
  \end{tabular}
\end{center}
Running entailment reduction on $T_3$ would remove $C_2$, $C_3$, and $C_4$ because they are logical consequence of $C_1$.
But it is impossible to learn the intended father relation given only $C_1$.
By contrast, running derivation reduction on $T_3$ would only remove $C_4$ because it can be \emph{derived} by self-resolving $C_3$.
As this example illustrates, any clause removed by derivation reduction can be recovered by derivation if necessary, while entailment reduction can be too strong and remove important clauses with no way to get them back using SLD-resolution.
In this paper, we address this issue by studying the derivation reducibility of fragments of second-order Horn logic relevant to ILP.
Although our notion of derivation reduction can be defined for any proof system, we initially focus on SLD-resolution because (1) most forms of ILP learn definite logic programs (typically Prolog programs), and (2) we want to reduce sets of metarules, which are themselves definite clauses (although second-order rather than first-order).
The logic fragments we consider here also correspond to the search spaces typically targeted by ILP systems.

\subsection{Contributions}
Our main contributions are:
\begin{itemize}
\item We state the derivation reduction problem for SLD-resolution (Sect.\ \ref{sec:prob}) that we originally introduced in \cite{cropper2018derivation}.
\item We describe fragments of second-order Horn logic particularly relevant for ILP (Sect.\ \ref{sec:fragments}).
\item We show that, by constraining the arity of the predicates, an infinite fragment of connected Horn clauses can be derivationally reduced to a finite fragment made of clauses that contain at most two literals in the body (Sect.\ \ref{sec:connmost}).
\item We show that an infinite fragment of 2-connected (i.e. connected and without singleton occurrences of variables) Horn clauses \emph{cannot} be derivationally reduced to any finite fragments (Sect.\ \ref{sec:chainmost}).
\item We show similar but incomplete negative results for a more expressive 2-connected fragment (Sect.\ \ref{sec:chainmost-three}).
\item We extend the reducibility results to standard resolution (Sect.\ \ref{sec:res}).
\end{itemize}
\conf{
  A technical report including detailed proofs of all the results (including the ones only sketched in this paper) has been created as a separate document \cite{techrepjelia2018tc}.
}

\section{Related Work}

In clausal logic there are two main forms of redundancy: (1) a literal may be redundant in a clause, and (2) a clause may be redundant in a clausal theory.

\vspace{-1\baselineskip}
\subsubsection{Literal Redundancy.}
Plotkin \cite{plotkin:thesis} used subsumption to decide whether a literal is redundant in a first-order clause.
Joyner \cite{DBLP:journals/jacm/Joyner76} independently studied the same problem, which he called \emph{clause condensation}, where a condensation of a clause $C$ is a minimum cardinality subset $C'$ of $C$ such that $C' \models C$.
Gottlob and Ferm\"uller \cite{gottlob1993removing} showed that determining whether a clause is condensed is coNP-complete.
In contrast to eliminating literals from clauses, we focus on removing {\em clauses} from theories.
\vspace{-1\baselineskip}
\subsubsection{Clause Redundancy.}
Plotkin \cite{plotkin:thesis} also introduced methods to decide whether a clause is subsumption redundant in a first-order clausal theory.
The same problem, and slight variants, has been extensively studied in the propositional logic \cite{liberatore1,liberatore2} and has numerous applications, such as to improve the efficiency of SAT solving \cite{heule2015clause}.
This problem has also been extensively studied in the context of first-order logic with equality due to its application in superposition-based theorem proving \cite{hillenbrand2013search,weidenbach2010subterm}.
Langlois et al.\ \cite{langlois2009combinatorial} studied combinatorial problems for propositional Horn clauses.
Their results include bounds on entailment reduced sets of propositional Horn fragments.
In contrast to these works, we focus on removing {\em second-order} Horn clauses (without equality) that are {\em derivationally} redundant.

Much closer to this paper is the work of Cropper and Muggleton \cite{crop:minmeta}.
They used entailment reduction \cite{mugg:progol} on sets of second-order Horn clauses to identify theories that are (1) entailment complete for certain fragments of second-order Horn logic, and (2) minimal or irreducible, in that no further reductions are possible.
They demonstrate that in some cases as few as two clauses are sufficient to entail an infinite language.

In contrast to all these works, we go beyond entailment reduction and introduce {\em derivation reduction} because, as stated in the previous section, the former can be too strong to be of use in ILP.
Thus our focus is on {\em derivationally} reducing sets of {\em second-order Horn} clauses.

\vspace{-1\baselineskip}
\subsubsection{Theory Minimisation and Program Transformation.}
In theory minimisation \cite{DBLP:conf/ijcai/HemaspaandraS11} the goal is to find a minimum equivalent formula to a given input formula.
The fold/unfold transformations of first-order rules are used, e.g.\ to improve the efficiency of logic programs or to synthesise definite programs from arbitrary specifications \cite{sato1992equivalence}.
Both allow for the introduction of new formul\ae.
By contrast, the derivation reduction problem only allows for the {\em removal} of redundant clauses.

\vspace{-1\baselineskip}
\subsubsection{Prime Implicates.}
Implicates of a theory $T$ are the clauses entailed by $T$.
They are called prime when they do not themselves entail other implicates of $T$.
This notion differs from the redundancy elimination in this paper because (1) the notion of a prime implicate has been studied only in propositional, first-order, and some modal logics \cite{marquis2000consequence,echenim2015quantifierfree,bienvenu2007prime}, and (2) implicates are defined using entailment, which as already stated is too strong for our purpose.

\vspace{-1\baselineskip}
\subsubsection{Descriptive Complexity.}
Second-order Horn logic is often the focus in descriptive complexity \cite{descriptive-complexity}, which studies how expressive a logic must be to describe a given formal language.
For instance, Gr\"adel showed that existential second-order Horn logic can describe all polynomial-time algorithms \cite{gradel1991expressive}.
In this paper, we do not study the expressiveness of the logic but whether the logic can be logically reduced.

\vspace{-1\baselineskip}
\subsubsection{Higher-Order Calculi.}
SLD-resolution on second-order clauses, as used in this paper, supports the unification of predicate variables.
By contrast, there are extensions of SLD-resolution and standard resolution that handle the full expressivity of higher-order logic \cite{charalambidis2013extensional,huet1973mechanization}. These richer extensions handle more complex clauses, e.g.\ clauses including function symbols and $\lambda$-terms. We do not consider such complex clauses because most ILP approaches use second-order Horn clauses to learn function-free first-order Horn programs \cite{mugg:metagold,inoue:mil-hex,crop:metafunc,evans2017learning,albarghouthi2017constraint}.
Extending our results to full higher-order logic is left for future work.

\vspace{-1\baselineskip}
\subsubsection{Second-Order Logic Templates.}
McCarthy \cite{DBLP:conf/mi/McCarthy95} and Lloyd \cite{lloyd:logiclearning} advocated using second-order logic to represent knowledge. Similarly, in \cite{mlj:ilp20}, the authors argued for using second-order representations in ILP to represent knowledge. As mentioned in the introduction, many forms of ILP use second-order Horn clauses as a form of declarative bias \cite{raedt:decbias} to denote the structure of rules that may be induced. However, most approaches either (1) assume correct templates as input, or (2) use clauses without any theoretical justifications. Recent work \cite{crop:minmeta} has attempted to address this issue by reasoning about the completeness of these templates, where the goal is to identify finite sets of templates sufficiently expressive to induce all logic programs in a given fragment. Our work contributes to this goal by exploring the derivation redundancy of sets of templates.

\vspace{-1\baselineskip}
\subsubsection{Derivation Reduction.}
In earlier work \cite{cropper2018derivation} we introduced the derivation reduction problem and a simple algorithm to compute reduction cores.
We also experimentally studied the effect of using derivationally reduced templates on ILP benchmarks.
Whereas our earlier paper mainly focuses on the application of derivation reduction to ILP, the current paper investigates derivation reduction itself in a broader perspective, with more emphasis on whether infinite fragments can be reduced to finite subsets.
Another main distinction between the two papers is that here we focus on derivation reduction modulo first-order variable unification.
The overlap includes the definition of derivation reduction and Sect.\ 4.2 in \cite{cropper2018derivation} which covers in less detail the same topic as our Sect.\ \ref{sec:chainmost}.

\section{Problem Statement and Decidability}
\label{sec:prob}

We now define the derivation reduction problem, i.e. the problem of removing derivationally redundant clauses from a clausal theory.

\subsection{Preliminaries}
We focus on function-free second-order Horn logic.
We assume infinite enumerable sets of {\em term variables} $\{x_1$, $x_2$, $\dots$\} and {\em predicate variables} \{$P$, $P_0$, $P_1$, $\dots$\}.
An {\em atom} $P(x_{k_1},\dots,x_{k_a})$ consists of a predicate variable $P$ of arity $a$ followed by $a$ term variables.
A {\em literal} is an atom ({\em positive} literal) or the negation of an atom ({\em negative} literal).
A {\em clause} is a finite disjunction of literals.
A {\em Horn clause} is a clause with at most one positive literal.
From this point on, we omit the term \emph{Horn} because all clauses in the rest of the paper are Horn clauses ($\lambda$-free function-free second-order Horn clauses to be precise).
The positive literal of a clause $C$, when it exists, is its {\em head} and is denoted as $h(C)$.
The set of negative literals of $C$ is called its {\em body} and is denoted as $b(C)$.
The clause $C$ is written as $h(C)\leftarrow b(C)$.
We denote the empty clause as $\Box$.
We denote the number of literals occurring in $b(C)$ as $\bsize{C}$, i.e.\ the body size of $C$.
A {\em theory} $T$ is a set of clauses.

A {\em substitution} $\sigma$ is a function mapping term variables to term variables, and predicate variables to predicate variables with the same arity.
The application of a substitution $\sigma$ to a clause $C$ is written $C\sigma$.
A substitution $\sigma$ is a {\em unifier} of two literals when they are equal after substitution.
A substitution $\sigma$ is a {\em most general unifier} of two literals, denoted as \mgu, when no smaller substitution is also a unifier of the two literals, i.e.\ there exist no $\sigma'$ and $\gamma$ such that $\sigma'$ unifies the two literals and $\sigma = \sigma'\circ\gamma$.
The variables in a clause are implicitly universally quantified. In practice, ILP approaches typically use existentially quantified predicate variables \cite{mugg:metagold,crop:minmeta,albarghouthi2017constraint,evans2017learning}. However, we ignore the quantification of the predicate variables because we are not concerned with the semantics of the clauses, only their syntactic form.

\subsection{Derivation Reduction}

The derivation reduction problem can be defined for any proof system but we focus on SLD-resolution \cite{sld-resolution} because of the direct application to ILP. SLD-resolution is a restricted form of resolution \cite{robinson:resolution} based on linear resolution with two main additional constraints (1) it is restricted to Horn clauses, and (2) it does not use factors, where factoring unifies two literals in the same clause during the application of the resolution inference rule (this implies that all resolvents are binary resolvents).
SLD-resolution is usually defined for first-order logic.
To apply it to the second-order clauses in this paper, we replace the standard notion of a m.g.u.\ with the one defined in the previous paragraph that also handles predicate variables.
An SLD-resolution inference is denoted as $C_1,C_2\vdash C$ where the necessary m.g.u.\ is implicitly applied on $C$.
The clauses $C_1$ and $C_2$ are the {\em premises} and $C$ is the {\em resolvent} of the inference.
The literal being resolved upon in $C_1$ and $C_2$ is called the {\em pivot} of the resolution.
We define a function $S^n(T)$ of a theory $T$ as:
\begin{center}
  \begin{tabular}{l}
    $S^0(T) = T$\\
    $S^n(T) = \{C  | C_1 \in S^{n-1}(T),C_2 \in T$, s.t. $C_1,C_2\vdash C\}$
  \end{tabular}
\end{center}
The {\em SLD-closure} of a theory $T$ is defined as: $$S^*(T) = \bigcup\limits_{n\in\mathbb{N}}S^n(T)$$
A clause $C$ is {\em derivable} from the theory $T$, written $T \vdash^* C$, if and only if $C \in S^*(T)$.
Given a theory $T$, a clause $C\in T$ is {\em reducible} if it is the resolvent of an inference whose premises all belong to $T$ and have a body size smaller than $\bsize{C}$.
A clause $C$ is {\em redundant} in the theory $T \cup \{C\}$ if and only if $T\vdash^* C$.
By extension, a theory $T$ is {\em redundant} to another theory $T'\subseteq T$ if for all $C\in T$, $T'\vdash^*C$.
A theory is {\em reduced} if and only if it does not contain any redundant clauses.
We state the \emph{reduction problem}:
\begin{definition}[Reduction Problem]
  \label{def:reduction-problem}
  Given a possibly infinite theory $T$, the reduction problem is to find a finite theory $T' \subseteq  T$ such that (1) $T$ is redundant to $T'$, and (2) $T'$ is reduced.
  In this case, we say that $T'$ is a \emph{reduction core} of $T$.
\end{definition}
Note that in the case of a finite theory $T$, the existence of a reduction core is obvious since at worst it is $T$ itself.
However, for arbitrary theories it is impossible to compute or a reduction core because the derivation reduction problem is undecidable \cite{cropper2018derivation}.

\section{Fragments of Interest in \Horn}
\label{sec:fragments}

From Sect.\ \ref{sec:connmost} onwards we study whether derivationally reduced theories exist for various fragments of Horn logic.
Horn logic with function symbols has the expressive power of Turing machines and is consequently undecidable \cite{tarnlund:hornclause}, hence ILP approaches typically learn programs without function symbols \cite{ilp:book}, which are decidable \cite{datsin:datalog}.
We therefore focus on function-free Horn clauses.
We denote the set of all second-order function-free Horn clauses as \Horn.

We further impose syntactic restrictions on clauses in \Horn\, principally on the arity of the literals and on the number of literals in the clauses.
Let us consider a fragment \F\ of \Horn.
We write $\F_{a,b}$ to denote clauses in \F\ that contain literals of arity at most $a$ and clauses of body size at most $b$.
For example, the clause $P_0(x_1)\leftarrow P_1(x_2,x_3,x_4)$ is in $\Horn_{3,1}$.
When one of these restrictions is not imposed, the symbol $\infty$ replaces the corresponding number.
When restrictions are imposed on a fragment that is already restricted, the stricter restrictions are kept.
For example, $(\Horn_{4,1})_{3,\infty}=\Horn_{3,1}=\Horn_{4,1}\cap\Horn_{3,\infty}$.
We rely on the body size restriction to bound the reduction cores of the studied fragments.

We also constrain the fragments so that they are defined modulo variable renaming and so that only the most general clauses up to variable unification are considered.
Let $C$ be a clause verifying the syntactic restrictions of a given fragment $\F$.
Then there exists a clause $C_\F\in \F$ such that $C_\F\sigma=C$ for some substitution $\sigma$.
The motivation behind this restriction is that SLD-resolution only applies \mgu s and not any unifiers but some clauses like $C'$ may need more specific unifiers to be generated and can thus be unreachable by SLD-resolution.
This is not restrictive because up to variable renaming any such $C'$ can be obtained from $C$ by renaming and unifying variables.
\begin{definition}[Reducible fragment]
  A fragment \F\ of \Horn\ is reducible to $\F_{\infty,b}$ when, for all $C\in\F$ such that $b<\bsize{C}$, there exists $b'<\bsize{C}$ such that $\F_{\infty,b'}\vdash C$, i.e.\ $C$ is the resolvent of an inference with premises in $\F_{\infty,b'}$.
\end{definition}
The following results are consequences of this definition and of the reduction problem statement.
\begin{proposition}[Reduciblility]
  \label{prop:freduc}
  If a fragment \F\ is reducible to $\F_{\infty,b}$ then \F\ is redundant to $\F_{\infty,b}$.
\end{proposition}

\begin{theorem}[Cores of Reducible Fragments]
  If a fragment \F\ is reducible to $\F_{\infty,b}$ then the solutions of the reduction problem for \F\ and $\F_{\infty,b}$ are the same, i.e. the reduction cores of \F\ and $\F_{\infty,b}$ are the same.
\end{theorem}

Because we are motivated by applications in ILP, we focus on connected clauses \cite{crop:minmeta,ilp:book,albarghouthi2017constraint,evans2017learning,inoue:mil-hex}:
\begin{definition}[Connected Fragment]
  \label{def:connected}
  A clause is connected if the literals in the clause cannot be partitioned into two non-empty sets such that the variables appearing in the literals of one set are disjoint from the variables appearing in the literals of the other set. The {\em connected} fragment, denoted as \conn, is the subset of \Horn\ where all clauses are connected.
\end{definition}
\begin{example}
  \label{ex:conn}
  The clause $C_1=P_0(x_1,x_2)\leftarrow P_1(x_3,x_1),$ $P_2(x_2),$ $P_3(x_3)$ is in \conn, but the clause $C_2=P_0(x_1,x_2)\leftarrow P_1(x_3,x_4),$ $P_2(x_2),P_3(x_3)$ is not because none of the variables in $P_0$ and $P_2$ ($x_1$ and $x_2$) appear in $P_1$ and $P_3$ and vice versa.
\end{example}

A stricter version of connectedness, denoted here as 2-connected\-ness, describes the fragment that is used the most in ILP \cite{crop:minmeta}.
It essentially eliminates singleton variables.
\begin{definition}[2-Connected Fragment]
  The {\em 2-connected} fragment, denoted as \chain, is the subset of \conn\ such that all the term variables occur at least twice in  distinct literals.
  In this context, a term variable that does not follow this restriction is denoted as {\em pending}.
\end{definition}
\begin{example}
  The clause $C_1$ from Example \ref{ex:conn} is in \chain\ because $x_1$ is in $P_0$ and $P_1$, $x_2$ is in $P_0$ and $P_2$, and $x_3$ is in $P_1$ and $P_3$.
  By contrast, the clause $C_3=P_0(x_1,x_2)\leftarrow P_1(x_3,x_1),P_2(x_1),P_3(x_3)$ is in \conn\ but not in \chain\ because $x_2$ only occurs once and is thus pending.
\end{example}
Note that the simple syntactic restrictions can be combined with both connectedness and 2-connectedness.
In the following sections we consider the reduction problem for \conn\ (Sect.\ \ref{sec:connmost}), $\chainmost{2}$ (Sect.\ \ref{sec:chainmost}), and $\chainmost{3}$ (Sect.\ \ref{sec:chainmost-three}).

\section{The Fragment \conn\ is Reducible to $\conn_{\infty,2}$}
\label{sec:connmost}

We now study whether certain fragments can be reduced.
Our first focus is on the fragment \conn which contains all connected clauses.
We are primarily interested in whether this fragment can be reduced using SLD-resolution to a minimal fragment, preferably with only two literals in the body ($\conn_{\infty,2}$).

\subsection{Graph Encoding}
To prove the reducibility of \conn{} we consider \connmost{a} for any $a\in\N^*$ and show that it can be reduced to $\conn_{a,2}$.
To reduce all clauses in \connmost{a} of body size greater than two, we rely on the following graph encoding to create connected premises to infer $C$.
We assume reader familiarity with basic notions of graph theory, in particular, notions of \techrep{{\em circuits} and {\em length} of circuits, }{\em spanning trees}, {\em connected graphs}, {\em degree} of vertices and {\em outgoing} edges (from a set of vertices).

\begin{definition}[Graph Encoding]
  \label{def:graph}
  Let $C$ be a clause in \connmost{m}.
  The undirected graph \Gr{C} is such that:
  \begin{itemize}
  \item There is a bijection between the vertices of \Gr{C} and the predicate variable occurrences in $C$ (head and body).
  \item There is an edge in \Gr{C} between each pair of vertices for each corresponding pair of literals that share a common term variable.
    The edge is labeled with the corresponding variable.
  \end{itemize}
\end{definition}
\begin{example}
  $C=P_0(x_1,x_2)\leftarrow$ $P_2(x_1,x_3,x_4),$ $P_3(x_4),$ $P_4(x_2,x_5),$ $P_1(x_5,x_6)$ is mapped to \Gr{C} as illustrated in Fig.\ \ref{fig:graph-example}.
  Note that since the variables $x_3$ and $x_6$ occur only in $P_2$ and $P_1$ respectively, they are not present in \Gr{C}.
  In fact \Gr{C} also represents many other clauses, e.g. $P_1(x_5,x_5)\leftarrow P_0(x_2,x_1),$ $P_2(x_4,x_3,x_1),P_3(x_4),$ $P_4(x_2,x_5)$.
\end{example}
\begin{figure}[t]
  \centering
  \includegraphics[scale=.9]{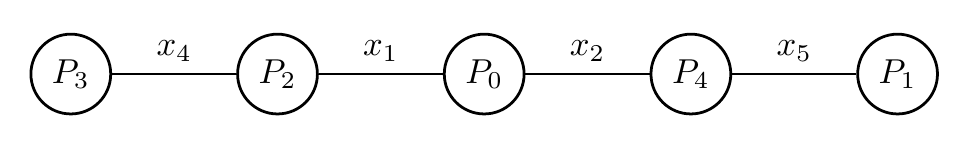}
\caption{Encoding of $C=P_0(x_1,x_2)\leftarrow$ $P_2(x_1,x_3,x_4),P_3(x_4),P_4(x_2,x_5),P_1(x_5,x_6)$ where vertices correspond to literals and edges represent variables shared by two literals}
\label{fig:graph-example}
\end{figure}

This graph encoding allows us to focus on connectivity, as stated in the following proposition.
\begin{proposition}
  Let $C\in\Horn$.
  The graph \Gr{C} is connected if and only if $C\in\conn$.
\end{proposition}
In other words, the notion of connectedness that we introduced for clauses in Def.\ \ref{def:connected} is equivalent to graph connectedness when encoding the clauses in graph form using Def.\ \ref{def:graph}.
Because we are only interested in connected clauses, we only handle connected graphs.

\subsection{Reducibility of \conn}

\techrep{The proofs in this section assume that in a clause (1) no two literals share more than one variable, and (2) no predicate variable occurs more than once.
Condition (1) allows us to identify edges with pairs of vertices without care for their label since there is then at most one edge between two vertices, as is the case in standard unlabeled graphs.
The clauses not verifying this condition are less general than the clauses that do.
For example, a connected clause $C$, containing the literals $P_1(x_1,x_2,\dots)$ and $P_2(x_1,x_2,\dots)$, can be obtained from the connected clause $C'$ equal to $C$, except that the literals $P_2(x_1,x_2,\dots)$ is replaced with $P_2(x_1,x_3,\dots)$ where $x_3$ is a variable not occurring in $C$, by using the substitution that maps $x_3$ to $x_2$.
Condition (2) allows us to name vertices with the predicates they represent.
Clauses that do not respect this criterion can also be produced by unifying variables.
These two constraints stem directly from our working modulo variable unification.}

Proposition \ref{prop:adj} is the main intermediary step in the proof of reducibility of the connected fragment (Th.\ \ref{th:connmostred}).
\conf{A detailed proof of this result is available in the technical report version of this paper \cite{techrepjelia2018tc}.}
\techrep{
The following theorem is a reminder of a classical result in graph theory that is used in the subsequent proof.
\begin{theorem}[Th.\ 2 from \cite{dirac1952some}]
  \label{\appdx th:dirac}
  A finite graph in which the degree of every vertex is at least $d(>1)$ contains a circuit of length at least $d+1$.
\end{theorem}

\begin{proposition}
  \label{\appdx prop:span}
  Let \Gr{} be a connected graph containing a circuit of length 3.
  If \Span\ is a spanning tree of \Gr{} containing two edges of the circuit then replacing in \Span\ one of these edges by the non-used one from the circuit yields another spanning tree of \Gr{}.
\end{proposition}
\begin{proof}
  Let \Span' be \Span\ after the replacement described in the proposition and $v_1$, $v_2$ and $v_3$ be the vertices in the circuit of length 3.
  We assume w.l.o.g.\ that the edges $(v_1,v_2)$ and $(v_2,v_3)$ belong to \Span\ and that the edge $(v_2,v_3)$ is replaced by $(v_1,v_3)$ in \Span'.
  Wherever there exists a path between two vertices in \Span, there also exists a path between them in \Span':
  \begin{itemize}
  \item if the path doesn't go through $(v_2,v_3)$ in \Span\ then the same path also exists in \Span',
  \item if the path goes through $(v_2,v_3)$ in \Span\ then the same path where $v_1$ is inserted between all contiguous occurrences of $v_2$ and $v_3$ exists in \Span'.
  \end{itemize}
  Let us assume the existence of a circuit in \Span'.
  Since \Span\ is a spanning tree, the circuit in \Span' necessarily go through $(v_1,v_3)$.
  Let us consider a path from $v_1$ to itself in \Span'.
  Then by inserting $v_2$ in any contiguous occurrences of $v_1$ and $v_3$, a path from $v_1$ to itself in \Span\ is obtained, a contradiction.
\end{proof}
}
\begin{proposition}[Spanning Tree]
  \label{\appdx prop:adj}
  For any clause $C\in\connmost{a}$, $a\in\N^*$, there exists a spanning tree of \Gr{C} in which there exist two adjacent vertices such that the number of edges outgoing from this pair of vertices is at most $a$.
\end{proposition}
\conf{
\begin{proofsk}
  Assuming no such pair of vertices exists in any spanning tree of \Gr{C}, we show in a case analysis that it is always possible to transform a spanning tree into another one where such a pair exists, a contradiction.
\end{proofsk}
}
\techrep{
\begin{proof}
  By contradiction, let us assume that no such pair of vertices exists in any spanning tree.
  Due to Th.\ \ref{\appdx th:dirac} there exists at least one vertex $v$ of degree 1 in the spanning tree, because by definition it has no circuit.
  The vertex $v'$ adjacent to $v$ is thus of degree at least $a+2$, so that there are at least $a+1$ outgoing edges from the pair $(v,v')$.
  Among the $a+2$ edges outgoing from $v'$, at least two are labeled with the
  same term variable(s) as some other edge(s).
  We want to remove one of these edges and replace it with an edge not connected to $v'$.
  These two edges may or may not be labeled with the same variable, thus there are two cases to examine:
  \begin{enumerate}
  \item there is one variable that is the label of at least three distinct edges outgoing from $v'$, thus connecting $v'$ with at least three other distinct vertices, or
  \item there are two distinct variables such that each one is the label of two distinct edges outgoing from $v'$, thus connecting two vertices to $v'$, the four such vertices being distinct.
  \end{enumerate}

  In the first case, since there are at least three distinct vertices connected to $v'$, at least two are distinct from $v$.
  We call $w$ and $w'$ these two vertices.
  Since both $w$ and $w'$ are connected to $v'$, there is no edge between them in the spanning tree, or it would have a circuit.
  However, note that the edge $(w,w')$ belongs to \Gr{C} because the corresponding literals share a common variable.
  Let us consider the same spanning tree where the edge $(w,v')$ has been replaced by the edge $(w,w')$.
  This new graph is also a spanning tree of \Gr{C} by Prop.\ \ref{\appdx prop:span}.

  In the second case, among the two pairs of vertices previously identified, we consider the pair of vertices $\{w,w'\}$ such that both are distinct from $v$.
  Since the four vertices are distinct, one such pair necessarily exists.
  As with the first case, since $w$, $w'$ and $v'$ share a common variable, it is possible to swap the edge $(w,v')$ with $(w,w')$ in the spanning tree and as in the first case, Prop.\ \ref{\appdx prop:span} guarantees that the obtained graph is also a spanning tree of \Gr{C}.

  In both cases, this operation reduces the number of edges outgoing from the pair $v$, $v'$ by one.
  This process can be repeated until this number reaches $a$, since the conditions to apply the transformation hold as long as the degree of $v'$ is greater than $a$.
  This contradicts our initial assumption.
\end{proof}
}

The main result of this section is the next theorem stating that any connected fragment of constrained arity has a reduction core containing clauses of body size at most two.

\begin{theorem}[Reducibility of \connmost{a}]
  \label{th:connmostred}
  For any $a\in\N^*$, \connmost{a} is reducible to $\conn_{a,2}$.
\end{theorem}
\begin{proof}
  Let $a\in\N^*$ be fixed and $C=P_0(\dots)\leftarrow P_1(\dots),\dots,P_k(\dots)\in\connmost{a}$ ($k\geq 3$).
  By applying Prop.\ \ref{prop:adj}, it is possible to identify two adjacent vertices $v$ and $v'$ in \Gr{C} such that there exists a spanning tree \Span\ of \Gr{C} where the number of edges outgoing from the pair $v$, $v'$ is less than or equal to $a$.
  Let $P_v$ and $P_{v'}$ be the predicate variables respectively corresponding to $v$ and $v'$ in $C$.
  Let $x_1,\dots,x_{a'}$ ($a'\leq a$) be the variables corresponding to the edges outgoing from the pair of vertices $v$, $v'$.
  Let $P_0'$ be an unused predicate variable of arity $a'$.
  We define: \(C_1 = P_0(\dots) \leftarrow P_0'(x_1,\dots,x_{a'}),P_1(\dots),\dots,P_k(\dots)\backslash\{P_v(\dots),P_{v'}(\dots)\}\) and \(C_2 = P_0'(x_1,\dots,x_{a'}) \leftarrow P_v(\dots),P_{v'}(\dots)\).
  These clauses are such that $C_1,C_2\in\conn_{a',\infty}$ and $C_1,C_2\vdash C$ modulo variable unification.\footnote{Some connections may be lost between variables in $C_1$ and $C_2$ since only the ones occurring in the spanning tree \Span\ are preserved. However, they can be recovered by unifying the disconnected variables together in the resolvent.}
  Thus, $C$ is reducible.
\end{proof}
We extend this result to the whole connected fragment.
\begin{theorem}[Reducibility of \conn]
\label{theorem14}
The fragment \conn\ is reducible to $\conn_{\infty,2}$.
\end{theorem}
Note that Theorem \ref{theorem14} does not imply that \conn\ has a reduction core because $\conn_{\infty,2}$ is also infinite.
In fact, since it is not possible to increase the arity of literals through SLD-resolution, any fragment where this arity is not constrained is guaranteed to have no reduction core since at least one literal of each arity must occur in it and the number of literals that occur in a clause is finite.

\section{Reducibility of \chainmost{2}}
\label{sec:chainmost}

We now consider the reducibility of \chainmost{2}. The restriction to monadic and dyadic literals is common not only in ILP \cite{mugg:metagold,evans2017learning,albarghouthi2017constraint,andreas} but also in description logics \cite{DBLP:books/daglib/0041477} and in ontology reasoning \cite{DBLP:journals/corr/HoheneckerL17}.
Although this fragment is only slightly more constrained than \connmost{2}, itself reducible to $\conn_{2,2}$, we show that it is impossible to reduce \chainmost{2} to any size-constrained sub-fragment.
To do so we exhibit a subset \nocut\ in \chainmost{2} that cannot be reduced.
This set contains clauses of arbitrary size.
In practice, this means that in \chainmost{2} given any integer $k$ it is possible to exhibit a clause of body size superior or equal to $k$ that cannot be reduced, thus preventing \chainmost{2} itself to be reducible to $\chain_{2,k}$ no matter how big $k$ is.
We start by defining the clause $\cbase\in\nocut$.
\begin{definition}[\cbase]
  \[
    \cbase=P_0(x_1,x_2)\leftarrow P_1(x_1,x_3),P_2(x_1,x_4),P_3(x_2,x_3),P_4(x_2,x_4),P_5(x_3,x_4).
  \]
\end{definition}
In \cbase\ all the literals are symmetrical to each other. Each literal (vertex) has (1) two neighbours connected by their first variable, (2) two other neighbours connected by their second variable, and (3) another literal that it is not connected to but which all the other literals are. This symmetry is better seen on the graphical representation of \cbase\ in Fig.\ \ref{fig:cbase}.
For example $P_0$ does not share literals with $P_5$ but does with all other predicates.
\begin{figure}[t]
  \centering
  \begin{subfigure}[b]{0.33\textwidth}
    \includegraphics[width=\textwidth]{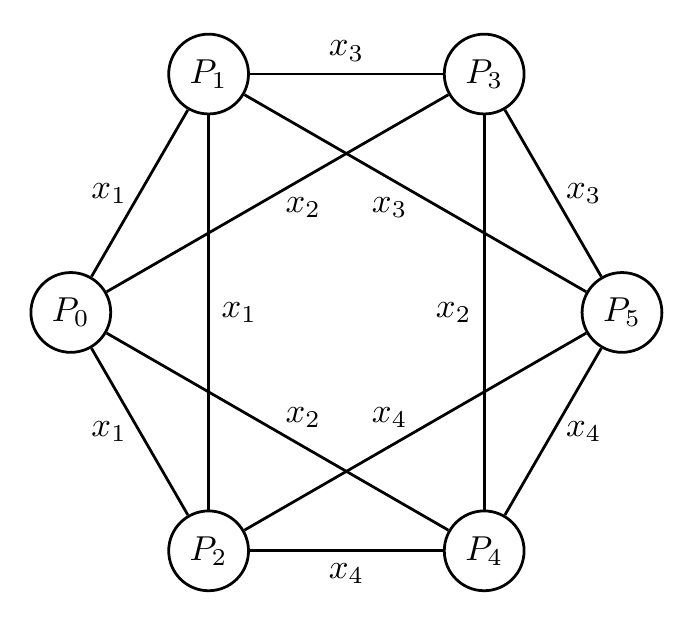}
    \caption{Graph encoding of \cbase, \Gr{\cbase}}
    \label{fig:cbase}
  \end{subfigure}
  ~
  \begin{subfigure}[b]{0.64\textwidth}
    \includegraphics[width=\textwidth]{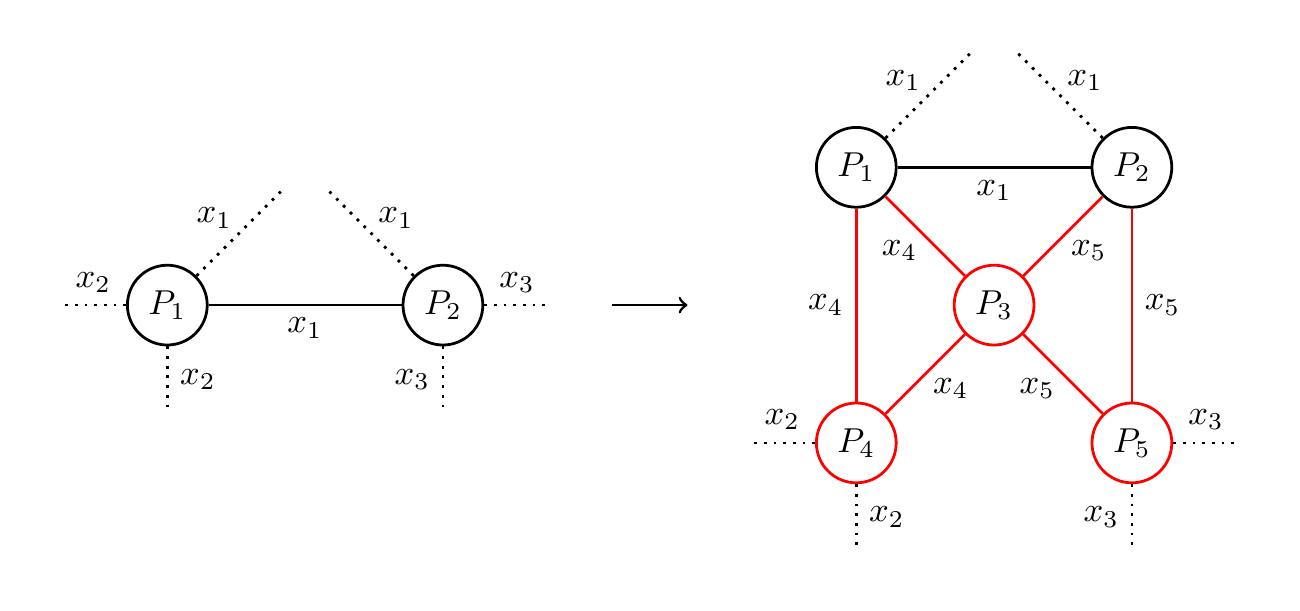}
    \caption{Partial graph encoding of a clause before and after a non-red preserving transformation}
    \label{fig:cext}
  \end{subfigure}
  \caption{Graph encoding of \nocut\ base and construction rule}
  \label{fig:nocut}
\end{figure}
\begin{proposition}[Non-reducibility of \cbase]
  \label{ijcarprop:base}
  $\cbase$ is irreducible.
\end{proposition}
\begin{proof}
  To derive \cbase\ from two smaller clauses, these two smaller clauses $C_1$ and $C_2$ must form a partition of the literals in \cbase\ if one excludes the pivot.
  To solve this problem, we partition the vertices of \Gr{\cbase} in two sets and count the number of edges with distinct labels that link vertices from the two sets.
  These edges correspond to pending variables in one of the sets, i.e.\ to the variables that must occur in the pivot that will be added in both sets to form $C_1$ and $C_2$.
  If there are more than two of these variables, the pivot cannot contain all of them, thus at least one of $C_1$ and $C_2$ is not in \chainmost{2} for lack of 2-connectivity.
  Each of the two sets in the partition must contain at least two elements, otherwise one of $C_1$, $C_2$ is as big as \cbase\ which does not make \cbase\ reducible even though it is derivable from $C_1,C_2$.
  The symmetries in \Gr{\cbase} are exploited to reduce the number of cases to consider to only four that vary along two dimensions: the cardinalities of the two subsets, either 2-4 or 3-3 respectively; and the connectedness of the subsets.
  In the 2-4 partition, only the following cases or symmetric ones are possible:
  \begin{itemize}
  \item if $\{P_0,P_5\}$ is the subset of cardinality 2 in a 2-4 partition, then the edges outgoing from this subset, connecting the two subsets and that correspond to pending variables, are labeled with $x_1$, $x_2$, $x_3$ and $x_4$;
  \item if $\{P_0,P_1\}$ is the subset of cardinality 2 in a 2-4 partition, then the outgoing edges are labeled with $x_1$, $x_2$ and $x_3$.
  \end{itemize}
  All the remaining 2-4 cases where $P_0$ is in the subset of cardinality 2 are symmetric to this case.
  The other 2-4 cases are symmetric to either one of these two cases.
  Similarly, all the 3-3 partition are symmetric to one of the following cases:
  \begin{itemize}
  \item if $\{P_0,P_1,P_2\}$ is one of the subsets in a 3-3 partition then the outgoing edges are labeled with $x_2$, $x_3$ and $x_4$;
  \item if $\{P_0,P_1,P_4\}$ is one of the subsets in a 3-3 partition then the outgoing edges are labeled with $x_1$, $x_2$, $x_3$ and $x_4$.
  \end{itemize}
  In all cases, there are 3 or more distinct labels on the edges between the two subsets, corresponding to pending variables, thus \cbase\ is irreducible.
  Note that this proof works because there are exactly three occurrences of each variable in \cbase.
  Otherwise it would not be possible to match the labels with the pending variables.
\end{proof}

We define a transformation that turns a clause into a bigger clause (Def.\ \ref{def:cext}) such that when applied to an irreducible clause verifying some syntactic property, the resulting clause is also irreducible (Prop.\ref{prop:incr}).
\begin{definition}[Non-red Preserving Extension]
  \label{def:cext}
  Let the body of a clause $C\in\chainmost{2}$ contain two dyadic literals sharing a common variable, e.g.\ $P_1(x_1,x_2)$ and $P_2(x_1,x_3)$, without loss of generality.
  A {\em non-red preserving} extension of $C$ is any transformation which replaces two such literals in $C$ by the following set of literals: $P_1(x_1,x_4)$, $P_2(x_1,x_5)$, $P_3(x_4,x_5)$, $P_4(x_4,x_2)$, $P_5(x_5,x_3)$ where $P_3$, $P_4$, $P_5$, $x_4$ and $x_5$ are new predicate and term variables.
\end{definition}

\techrep{On Fig.\ \ref{fig:cext} the neighboring structure of the literals after the non-red preserving transformation shows a symmetry between the ordered pairs of vertices $(P_1,P_4)$ and $(P_2,P_5)$ that mirrors the symmetry between the original vertices $P_1$ and $P_2$.
Even if this symmetry is partial due to the fact that the rest of the clause is unknown, it can still be exploited to reduce the number of cases to consider in the proof of Prop.\ \ref{prop:incr}}
\noindent
\begin{proposition}[Non-red Preserving Extension]
  \label{\appdx prop:incr}
  If a clause $C$ is irreducible and all the term variables it contains occur three times then any non-red preserving extension of $C$ is also irreducible.
\end{proposition}
\conf{
  \begin{proofsk}
    We assume that a non-red preserving extension of $C$ is reducible and we use a case analysis to show that this implies that $C$ is also reducible, a contradiction.
    This proof heavily exploits the symmetry that can be seen on Fig.\ \ref{fig:cext} to reduce the number of cases to consider.
  \end{proofsk}
}
\techrep{
\begin{proof}
 Let $C$ be a irreducible clause containing the two literals $P_1(x_1,x_2)$ and $P_2(x_1,x_3)$ (without loss of generality) in which all variables occur exactly three times.
 Let \cext{} be the non-red preserving extension of $C$ where the two previously mentioned literals have been replaced by the set of literals $P_1(x_1,x_4)$, $P_2(x_1,x_5)$, $P_3(x_4,x_5)$, $P_4(x_4,x_2)$, $P_5(x_5,x_3)$ where $P_3$, $P_4$, $P_5$, $x_4$ and $x_5$ are new predicate and term variables.
 Assume that \cext{} is reducible.
 Then there exist two clauses \cext{1} and \cext{2} in \chainmost{2} both smaller than \cext, such that $\cext{1},\cext{2}\vdash \cext{}$.
 If \cext{1} is made of a subset of the literals in $\cext{}\backslash C$ (plus a pivot), then \cext{1} is not 2-connected because all these subsets leave three or more variables pending.
 The pending variables for each case are described in Tab.\ \ref{\appdx tab:fullcut} (the symmetrical cases are excluded).
 To illustrate how the table was built, we consider the case where \cext{1} contains $P_1(x_1,x_4),P_2(x_1,x_5),$ $P_3(x_4,x_5),P_4(x_4,x_2)$, i.e.\ the second line of Tab.\ \ref{\appdx tab:fullcut}.
 Consider these four literals, the variable $x_2$ is pending since it occurs exactly once.
 In addition, since the variables $x_1$ and $x_5$ occur only three times in \cext{}, they are also pending, albeit in \cext{2}.
 In the table, variables that are pending for this reason are followed by a star ($\star$).
 In total, there are three variables pending, which is one too many for the pivot to include all of them as arguments. 
 \begin{table}[t]
   \centering
   \caption{Pending variables when the given literal set is the body of \cext{1} - the $\star$ symbol indicates a variable pending in \cext{2}}
   \label{\appdx tab:fullcut}
   \begin{tabular}{|l|l|}
    \hline
    new literals in \cext{1} & pending variables \\
    \hline
     $P_1(x_1,x_4),P_2(x_1,x_5),P_3(x_4,x_5),$ &\\ \,$P_4(x_4,x_2),P_5(x_5,x_3)$ & $x_1\star,x_2\star,x_3\star$ \\
     $P_1(x_1,x_4),P_2(x_1,x_5),$ &\\ \,$P_3(x_4,x_5),P_4(x_4,x_2)$ & $x_1\star,x_2,x_5\star$ \\
     $P_1(x_1,x_4),P_2(x_1,x_5),$ &\\ \,$P_4(x_4,x_2),P_5(x_5,x_3)$ & $x_1\star,x_2,x_3,x_4\star,x_5\star$ \\
     $P_1(x_1,x_4),P_3(x_4,x_5),$ &\\ \,$P_4(x_4,x_2),P_5(x_5,x_3)$ & $x_1,x_2,x_3,x_5\star$ \\
     $P_1(x_1,x_4),P_2(x_1,x_5),P_3(x_4,x_5)$ & $x_1\star,x_4\star,x_5\star$ \\
     $P_1(x_1,x_4),P_2(x_1,x_5),P_4(x_4,x_2)$ & $x_1\star,x_2,x_4\star,x_5$ \\
     $P_1(x_1,x_4),P_3(x_4,x_5),P_4(x_4,x_2)$ & $x_1,x_2,x_5$ \\
     $P_1(x_1,x_4),P_3(x_4,x_5),P_5(x_5,x_3)$ & $x_1,x_3,x_4\star,x_5\star$ \\
     $P_1(x_1,x_4),P_4(x_4,x_2),P_5(x_5,x_3)$ & $x_1,x_2,x_3,x_4\star,x_5$ \\
     $P_3(x_4,x_5),P_4(x_4,x_2),P_5(x_5,x_3)$ & $x_2,x_3,x_4\star,x_5\star$ \\
     $P_1(x_1,x_4),P_2(x_1,x_5)$ & $x_1\star,x_4,x_5$ \\
     $P_1(x_1,x_4),P_3(x_4,x_5)$ & $x_1,x_4\star,x_5$ \\
     $P_1(x_1,x_4),P_4(x_4,x_2)$ & $x_1,x_2,x_4\star$ \\
     $P_1(x_1,x_4),P_5(x_5,x_3)$ & $x_1,x_3,x_4,x_5$ \\
     $P_3(x_4,x_5),P_4(x_4,x_2)$ & $x_2,x_4\star,x_5$ \\
     $P_4(x_4,x_2),P_5(x_5,x_3)$ & $x_2,x_3,x_4,x_5$ \\
     \hline
   \end{tabular}
    \end{table}
 The cases where \cext{2} is made only of the literals in $\cext{}\backslash C$ plus the pivot are symmetrical to the ones in Tab.\ \ref{\appdx tab:fullcut}.

 The remaining possibilities are when both \cext{1} and \cext{2} are made of a mix of the literals in $\cext{}\backslash C$ and $\cext{}\cap C$.
 In these cases, the contradiction appears by going from $\cext{1},\cext{2}\vdash\cext{}$ to $C_1,C_2\vdash C$.
 For example, if $P_1(x_1,x_4)$ and $P_2(x_1,x_5)$ belong to \cext{1} while the other literals from $\cext{}\backslash C$ belong to \cext{2}, then $x_4$ and $x_5$ are pending in \cext{1} (without pivot).
 There cannot be more than two variables pending in the pivot-less \cext{1}, \cext{2} pair or the pivot cannot take all of them as arguments so that they are not pending in $\cext{1}$ and $\cext{2}$ (with pivot), thus $x_4$ and $x_5$ are the only ones.
 Now consider $C_1$ and $C_2$, obtained respectively from \cext{1} and \cext{2} by deleting the five literals of $\cext{}\backslash C$ from them and adding $P_1(x_1,x_2)$ and $P_2(x_1,x_3)$, i.e.\ the literals in $C\backslash\cext{}$ into $C_1$.
 Before this transformation, the three occurrences of the variables $x_2$ and $x_3$ were located in \cext{2}.
 Due to the deletion of literals, only two occurrences of each remain in $C_2$ and one occurrence of each is now in $C_1$.
 Hence both $x_2$ and $x_3$ are pending in that case.
 Except for the variables $x_4$ and $x_5$ that are absent from $C_1$, $C_2$, the distribution of the remaining variables is unchanged when transforming \cext{1}, \cext{2} in $C_1$, $C_2$, hence these variables are not pending.
 Thus the pair $C_1$ $C_2$ make $C$ reducible, a contradiction.

 \begin{table}[t]
   \centering
   \caption{Transformation from $(\cext{1},\cext{2})$ to $(C_1,C_2)$ and corresponding evolution of the pending variables}
   \label{\appdx tab:partialcut}
   \begin{tabular}{|r@{\hspace{1mm}}c@{\hspace{1mm}}l|r@{\hspace{1mm}}c@{\hspace{1mm}}l|r@{\hspace{1mm}}c@{\hspace{1mm}}l|}
     \hline
     $\cext{1}$ & ; & $\cext{2}$ & \multicolumn{3}{c|}{pending variables} & $C_1$ & ; & $ C_2$ \\
     \hline
     $1,2,3,4,5$ & ; & $\emptyset$ & $\emptyset$ & ; & $\emptyset$ & $1,2$ & ; & $\emptyset$ \\
     $1,2,3,4$ & ; & $ 5$ & $x_5$ & ; & $ x_1$ & $1$ & ; & $ 2$ \\
     $1,2,4,5$ & ; & $ 3$ & $x_4,x_5$ & ; & $\emptyset$ & $1,2$ & ; & $\emptyset$ \\
     $1,3,4,5$ & ; & $ 2$ & $x_1,x_5$ & ; & $\emptyset$ & $1,2$ & ; & $\emptyset$ \\
     $1,2$ & ; & $ 3,4,5$ & $x_4,x_5$ & ; & $ x_2,x_3$ & $1,2$ & ; & $\emptyset$ \\
     $1,3$ & ; & $ 2,4,5$ & $x_1,x_4,x_5$ & ; & $\ast\ast\ast$ & $\ast\ast\ast$ & ; & $\ast\ast\ast$ \\
     $1,4$ & ; & $ 2,3,5$ & $x_4$ & ; & $\emptyset$ & $1$ & ; & $ 2$ \\
     $1,5$ & ; & $ 2,3,4$ & $x_1,x_4,x_5$ & ; & $\ast\ast\ast$ & $\ast\ast\ast$ & ; & $\ast\ast\ast$ \\
     $3,4$ & ; & $ 1,2,5$ & $x_4,x_5$ & ; & $ x_2$ & $\emptyset$ & ; & $ 1,2$ \\
     $4,5$ & ; & $ 1,2,3$ & $x_4,x_5$ & ; & $ x_2,x_3$ & $\emptyset$ & ; & $ 1,2$ \\
     \hline
   \end{tabular}
   \end{table}
 By taking into account all the symmetries of the problem, there are only ten such cases to consider.
 They are summarized in Tab.\ \ref{\appdx tab:partialcut}.
 On the left-hand side of the table is the partition between \cext{1} and \cext{2} of the five literals in $\cext{}\backslash C$.
 On the right-hand side of the table is the partition between $C_1$ and $C_2$ of the two literals in $C\backslash\cext{}$.
 As was done in the previous example, $C_1$ and $C_2$ are obtained by removing the five literals in $\cext{}\backslash C$ from \cext{1} and \cext{2} respectively and replacing them with the two literals in $C\backslash\cext{}$ as indicated in the table.
 For readability, the literals are only referred to by their number.
 In the middle of the table are the variables that are known to be pending in each case (in \cext{1} and \cext{2} on the left-hand side and in $C_1$ and $C_2$ on the right-hand side).
 In the cases where there are strictly less than two identified pending variables, there may also be unknown pending variables, but these are preserved by the transformation and thus do not impact the reasoning.
 In most of the cases, it is possible to have at most two variables pending on the right-hand side of the table, implying that $C$ is reducible, a contradiction.
 There are also two cases where the assumption that $\cext{1},\cext{2}\in\chainmost{2}$ is not verified because there are already more than two variables pending in the explicit parts of \cext{1} and \cext{2}.
 In such cases, there is nothing to verify so the right-hand side of the table is filled with asterisks (*).
\end{proof}
}

Starting from \cbase\ and using this extension, we define \nocut\ formally (Def.\ \ref{def:nocut}) and, as a consequence of Prop.\ \ref{prop:incr}, \nocut\ contains only irreducible clauses (Prop.\ \ref{prop:nocut}).

\begin{definition}[Non-reducible Fragment]
  \label{def:nocut}
  The subset \nocut\ of \chainmost{2} contains \cbase\ and all the clauses that can be obtained by applying a non-red extension to another clause in \nocut.
\end{definition}

\begin{proposition}[Non-reducibility of \nocut]
  \label{prop:nocut}
  For all $C\in\nocut$, $C$ is irreducible.
\end{proposition}
The non-reducibility of \nocut\ ensures that the body size of the clauses in a hypothetical reduction core of $\chainmost{2}$ cannot be bounded, which in turn prevents the existence of this reduction core.
This result has negative consequences on ILP approaches that use second-order templates.
We discuss these consequences in the conclusion.

\section{Reducibility of \chainmost{3}}
\label{sec:chainmost-three}
\begin{figure}[t]
  \centering
  \includegraphics[scale=0.7]{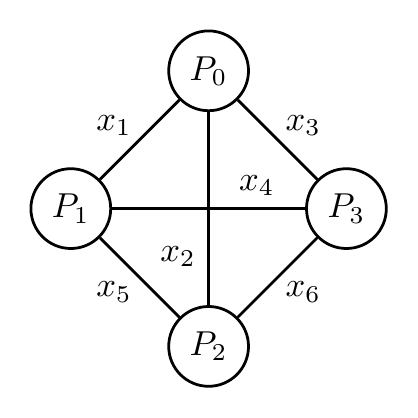}
  \caption{\Gr{C} for $C=P_0(x_1,x_2,x_3)\leftarrow P_1(x_1,x_4,x_5),$ $P_2(x_2,x_5,x_6),P_3(x_3,x_4,x_6)$}
  \label{fig:cnored}
\end{figure}

The reducibility of \chainmost{3} is still an open problem.
However, we know that it cannot be reduced to $\chain_{3,2}$.
\begin{theorem}[Non-reducibility of $\chain_{3,2}$]
  \label{th:notriad}
  \chainmost{3} cannot be reduced to $\chain_{3,2}$
\end{theorem}
\begin{proof}
  The clause $C=P_0(x_1,x_2,x_3)\leftarrow P_1(x_1,x_4,x_5),$ $P_2(x_2,x_5,x_6),$ $P_3(x_3,$ $x_4,x_6)$, shown in graph form in Fig.\ \ref{fig:cnored}, is a counter-example because any pair of literals in it contain exactly four pending variables.
  For example, consider the following pair of literals: $(P_1(x_1,x_4,x_5),P_0(x_1,x_2,x_3))$ leaves $x_2,x_3,x_4,x_5$ pending.
By symmetry the same holds for all the other pairs of literals.
Thus none of these pairs can be completed by a triadic (or less) pivot.
In addition, the removal of any single literal from $C$ does not lead to a reduction of the clause since all the variables occurring in the literal then occur only once in each subset of the clause.
For example, to replace $P_1(x_1,x_4,x_5)$, a triadic literal containing $x_1$, $x_4$ and $x_5$ needs to be added, creating a clause identical to $C$ up to the name of one predicate variable and the order of the term variables in it.
Therefore $C$ is irreducible in \chainmost{3}, thus \chainmost{3} cannot be reduced to $\chain_{3,2}$.
\end{proof}
In addition to this result, for lack of finding a reduction to
$P_0(x_1,x_2,x_3)\leftarrow$ $P_1(x_1,x_5,x_6),$ $P_2(x_2,x_4,x_8),$ $P_3(x_6,x_7,x_8),$ $P_4(x_4,x_5,$ $x_7),$ $P_5(x_3,x_4,x_7)$
(not formally proved) we conjecture that \chainmost{3} cannot be reduced to $\chain_{3,4}$.
Clarifying this situation and that of any \chainmost{a} with $a\geq 3$ is left as future work.

\section{Extension to Standard Resolution}
\label{sec:res}

Although we introduced the derivation reduction problem for SLD-resolution, the principle applies to any standard deductive proof system, and in particular, it can be applied to standard resolution, extended from first to second-order logic in the same way that was used for SLD-resolution.
Given that SLD-resolution is but a restriction of resolution, the positive reducibility result for $\connmost{2}$ is directly transferable to standard resolution.
On the contrary, the fragment $\chainmost{2}$, that we proved irreducible with SLD-resolution, can be reduced to $\chain_{2,2}$ with standard resolution.
\techrep{We identify with $_\res$ the notions where standard resolution replaces SLD-resolution.}
\begin{theorem}[Reducibility$_\res$ of \chainmost{2}]
  \label{\appdx th:chainmosttwo}
  \chainmost{2} is reducible$_\res$ to $\chain_{2,2}$
\end{theorem}
\techrep{
\begin{proof}
  Let $C\in\chainmost{2}$ such that $\bsize{C}\geq 3$.
  \begin{itemize}
  \item If $C$ contains a monadic literal, denoted $P(x)$, then due to the 2-connected constraint, there is at least another occurrence of $x$ in $C$.
    Let us denote the other literal in which $x$ occurs as $P_x$.
    Its parameters are left implicit since $P_x$ can be either monadic or dyadic, but $x$ is necessarily  one of them.
    The partition of $C$ into $\{P(x),P_x\}$, $C\backslash\{P(x),P_x\}$ leaves between zero and two variables pending (these pending variables occur in $P_x$).
    Thus $C$ is reducible$_\res$ and the premises of the corresponding inference are $C_1$ and $C_2$, both in \chainmost{2}, such that $\{P(x),P_x\}\subset C_1$, $C\backslash\{P(x),P_x\}\subset C_2$ and the pivot contains the pending variables.
    If no variable is pending, then $x$ occurs at least twice in both sets due to 2-connectedness, thus the resolving literal can also contain this variable while preserving the 2-connectedness of the two newly formed clauses.
    This transformation creates one rule of body size two ($C_1$), and another of body size $\bsize{C}-1$ ($C_2$).
  \item If $C$ contains two dyadic predicates with the same variables (the occurrence of two monadic predicates is covered by the previous case), e.g., $P_1(x_1,x_2)$ and $P_2(x_1,x_2)$, then the partition of $C$ into $\{P_1(x_1,x_2),P_2(x_1,x_2)\}$ and $C\backslash\{P_1(x_1,x_2),P_2(x_1,x_2)\}$ can be used in the same way as in the preceding case to obtain two smaller clauses resolving into $C$ since at most $x_1$ and $x_2$ are pending.
  \item If $C$ contains exactly two occurrences of the same variable in distinct literals or more than three occurrences of the same variable in distinct literals, it is possible to reduce$_\res$ $C$ by taking away two of these literals.
    This does not create more than two pending variables in the resulting partition, as in the two previous cases.
  \end{itemize}
  Let us now consider a clause $C$ that cannot be reduced$_\res$ following any of the three previous schemes.
  Such a clause $C$ has the following characteristics:
  \begin{itemize}
  \item $C$ contains only dyadic predicates,
  \item $C$ does not contain two predicates that have the same pair of variables,
  \item all variables in $C$ occur exactly three times.
  \end{itemize}
  If any of these conditions is not verified, one of the previous schemes can be applied on $C$ to reduce$_\res$ it as described.
  A first notable consequence of these characteristics is that $C$ is such that $\bsize{C}\geq 5$ because $C$ needs at least four variables to not have the same pair of variables occurring twice in different literals of $C$ (all variables occur three times, thus $C$ needs at least twelve variable occurrences grouped in 6 pairs).
  Of course, these conditions are not sufficient to prevent $C$ from being reducible$_\res$ in \chainmost{2}.

  To prove this point, let us assume that the clause $C$ is not reducible$_\res$ in \chainmost{2}.
  Let $x_1$ be a variable occurring in $C$ in the literals $P_1(x_1,x_2)$, $P_2(x_1,x_3)$ and $P_3(x_1,x_4)$ (without loss of generality).
  The partition of $C$ into $C'=\{P_1(x_1,x_2),$    $P_2(x_1,x_3), P_3(x_1,x_4)\}$ and $C\backslash C'$ leaves three variables pending in $C'$, namely $x_2$, $x_3$ and $x_4$, and none in $C\backslash C'$.
  As such, it is not suitable because one too many variable is pending.
  However, it is still possible to use this partition to find resolvents that prove $C$ reducible$_\res$ in \chainmost{2}.
  To do so, let us use again briefly the graph encoding from Def.\ \ref{def:graph}.
  It is possible to find a path in $\Gr{C\backslash C'}$ between two vertices mapped to predicates in each of which a distinct variable among $x_2$, $x_3$ and $x_4$ occurs.
  This is proven by contradiction.
  Assume no such path exists, then $\Gr{C\backslash C'}$ is made of three disconnected components, corresponding to three clauses $C_2$, $C_3$ and $C_4$, subclauses of $C$ where respectively only $x_2$, $x_3$ and $x_4$ occurs, i.e.\ in $C_2$, $x_3$ and $x_4$ do not occur but other unrelated variables possibly do, and $x_2$ certainly occurs, and $C_3$ and $C_4$ follow the same pattern.
  It follows that the partition of $C$ in, e.g., $C_2\cup\{P_1(x_1,x_2)\}$, $C_3\cup C_4\cup\{P_2(x_1,x_3),P_3(x_1,x_4)\}$, with the addition of a pivot on $x_1$ creates the necessary premises for $C$ to be reducible$_\res$, since $x_1$ is the only variable pending in the partition, a contradiction.
  Thus a path containing edges labeled with two of the variables $x_2$, $x_3$ and $x_4$ must exist in $\Gr{C\backslash C'}$.

  We assume w.l.o.g.\ that this path links vertices corresponding to predicates that occur applied respectively to $x_2$ and $x_3$ without going through an edge labeled with $x_4$ (otherwise, the shorter path between, e.g., the predicates that occur applied to $x_2$ and $x_4$ should be preferred).
  We denote $C_{23}$ the set of literals that are mapped in $\Gr{C\backslash C'}$ to edges in this path.
  Then the clause $C'\cup C_{23}$ has only one pending variable: $x_4$.
  By adding the pivot to the pair of clauses $C'\cup C_{23}$, $C\backslash C'$, premises to derive$_\res$ $C$ are created.
  Note that the situation where the head of $C$ occurs in both sets, preventing the addition of the resolving literal can be solved by selecting $x_1$ in such a way that the head of $C$ belongs to $C'$.
  The size of both sets is smaller than that of $C$ but still greater or equal to 3 since $|C|\geq 6$:
  \begin{itemize}
  \item $|C\backslash C'|=|C|-3$,
  \item $|C\cup C_{23}|\leq |C|-2$ because, without loss of generality, the path including edges labeled with $x_2$ and $x_3$ does not go through $x_4$ (in case it does, the shorter path where $x_2$ and $x_4$ occur as edge labels should be used instead), thus two of the three occurrences of $x_4$ in $C$ do not occur in $C\cup C_{23}$.
  \end{itemize}
  Thus the two obtained clauses are smaller than $C$ and we can conclude that $C$ is reducible$_\res$.
\end{proof}
}
\conf{
\begin{proofsk}
  We first analyse the structure of $C$ and show how to reduce $C$ in the simple cases where it is also possible to reduce $C$ using SLD-resolution.
  We are then left to consider clauses where $C$ contains only dyadic predicates, no two predicates in $C$ have the same pair of variables and all variables occur exactly three times in $C$.
  An example of such clauses is the \nocut\ family from Sect.\ \ref{sec:chainmost}.
  Then we present a method to reduce$_\res$ such a clause $C$.
  The key point that justifies Th.\ \ref{th:chainmosttwo} is that in standard resolution, factorisation is allowed and thus allows inferences that remove duplicate literals.
  The removal of duplicate literals would be also possible with SLD-resolution but only when the fragment contains bodyless clauses which is prevented by 2-connectedness.
\end{proofsk}
}

Let us consider an example of additional inferences allowed with resolution but not with SLD-resolution in the \chainmost{2} fragment, that make the C\(_{\text{base}}\) clause redundant:
\begin{center}
\begin{tabular}{rcl}
$P_0(x_1,x_2)$ & $\leftarrow$ & $P_1(x_1,x_3),P_2(x_1,x_4),P_3(x_2,x_3),H(x_2,x_4)$\\
$H'(x_2',x_4')$ & $\leftarrow$ & $P_3'(x_2',x_3'),P_4'(x_2',x_4'),P_5'(x_3',x_4')$\\
\hline
$P_0(x_1,x_2)$ & $\leftarrow$ & $P_1(x_1,x_3),P_2(x_1,x_4),P_3(x_2,x_3),P_3'(x_2,x_3'),P_4'(x_2,x_4),P_5'(x_3',x_4)$\\
\hline
$P_0(x_1,x_2)$ & $\leftarrow$ & $P_1(x_1,x_3),P_2(x_1,x_4),P_3(x_2,x_3),P_4'(x_2,x_4),P_5'(x_3,x_4)$\\
\end{tabular}
\end{center}
The first step is a resolution that unifies $H'$ with $H$, $x_2'$ with $x_2$ and $x_4'$ with $x_4$ and uses $H(x_2,x_4)$ as pivot.
The second step is a factorisation that unifies $P_3'$ with $P_3$, and $x_3'$ with $x_3$.
The result is \cbase\ up to variable renaming.

Finally, the result that we presented for $\chainmost{3}$ is also transferable from SLD- to standard resolution since the proof of Th.\ \ref{th:notriad} remains the same.
This is because the size of the considered clauses does not allow for the kind of resolution inferences that make Th.\ \ref{th:chainmosttwo} possible.
Table \ref{tab:sum} summarises our findings and their extension to standard resolution.

\begin{table}[bt]
  \centering
  \caption{Summary of the results.
    When a fragment is preceded with $>$ the entry must be read as ``no reduction up to this fragment''.
    The word \emph{possibly} precedes results that have not been proved and are only conjectured.}
  \label{tab:sum}
  \begin{tabular}{cr|r|}
    \cline{2-3}
    &\multicolumn{2}{|c|}{Reducibility} \\
    \hline
    \multicolumn{1}{|c|}{Fragment} & SLD-resolution & Standard resolution\\
    \hline
    \multicolumn{1}{|c|}{\conn} & $\conn_{\infty,2}$ & $\conn_{\infty,2}$ \\
    \multicolumn{1}{|c|}{\chainmost{2}} & no & $\chain_{2,2}$ \\
    \multicolumn{1}{|c|}{\chainmost{3}} & $>\chain_{3,2}$ & $>\chain_{3,2}$ \\
    \multicolumn{1}{|c|}{} & possibly $>\chain_{3,4}$ & possibly $>\chain_{3,4}$\\
    \hline
  \end{tabular}
  \end{table}

\section{Conclusion}
\label{sec:concl}

We have introduced the derivation reduction problem for second-order Horn clauses (\Horn), i.e.\ the undecidable problem of finding a finite subset of a set of clauses from which the whole set can be derived using SLD-resolution.
We have considered the derivation reducibility of several fragments of \Horn, for which the results are summarised in Tab.\ \ref{tab:sum}.
We have also extended the results from SLD-resolution to standard resolution.
Further work is necessary to clarify the situation for \chainmost{3} and for fragments with higher arity constraints.

Although we have positive results regarding the reducibility of certain fragments, we have not identified the reductions of those fragments, nor have we provided any results regarding the cardinality of the reductions.
Future work should address this limitation by introducing algorithms to compute the reductions.

Our results have direct implications in ILP. As described in the introduction, many ILP systems use second-order Horn clauses as templates to define the hypothesis space. An open question \cite{crop:minmeta,crop:thesis,mugg:metagold} is whether there exists finite sets of such clauses from which these systems could induce any logic program in a specific fragment of logic. Prop.\ \ref{prop:nocut} shows that for the \chainmost{2} fragment, which is often the focus of ILP, the answer is no. This result implies that ILP systems, such as Metagol \cite{metagol} and HEXMIL \cite{inoue:mil-hex}, are incomplete in that they cannot learn all programs in this fragment without being given an infinite set of clauses (these approaches require a finite set of such clauses hence the incompleteness).

Our work now opens up a new challenge of overcoming this negative result for \chainmost{2} (and negative conjectures for \chainmost{3}). One possible solution would be to allow the use of triadic literals as pivot in inferences in specific cases where SLD-resolution fails to derive the desired clause, but this idea requires further investigation.

\subsubsection{Acknowledgements.}
The authors thank Katsumi Inoue and Stephen Muggleton for discussions on this work.

{\small
\bibliography{jelia19_tourret_cropper}
\bibliographystyle{abbrv}
}
\end{document}